\documentclass[pre,aps,amsfonts,showpacs]{revtex4}
\usepackage{graphicx}
\usepackage{amsthm}
\usepackage{amssymb}
\usepackage{mathrsfs}
\theoremstyle{break}
\newtheorem{Theorem}{Theorem}[section]
\newtheorem{Corollary}{Corollary}[section]
\theoremstyle{change}

\begin{document}
\title{A generalization of the Kullback-Leibler divergence and its properties}
\author{{\sc Takuya Yamano}}
\email[Email: ]{yamano@amy.hi-ho.ne.jp}
\affiliation{Department of Physics, Ochanomizu University, 2-1-1 Otsuka, Bunkyo-ku, Tokyo 112-8610, Japan}

\begin{abstract}
A generalized Kullback-Leibler relative entropy is introduced starting with 
the symmetric Jackson derivative of the generalized overlap between two probability 
distributions. The generalization retains much of the structure possessed by the original formulation. 
We present the fundamental properties including positivity, metricity, concavity, bounds and stability. 
In addition, a connection to shift information and behavior under Liouville dynamics are discussed.
\medskip
\end{abstract}
\maketitle
\bigskip


\section{Introduction}
The relative entropy or the information divergence is a measure of the extent to which the assumed 
probability distribution deviates from the true one. We need a means of comparing 
two different probability distributions and will define a distance as a fundamental 
quantity that discriminates the distributions. The form of the relative entropy, which was first introduced 
by Kullback and Leibler \cite{KL} is the most pervasive measure in information theory \cite{Khinchin1} and 
statistical mechanics \cite{Khinchin2}. Its most prominent property lies in its asymmetry between two distributions 
(i.e., under interchange of the two) and that it does not satisfy the triangle inequality \cite{Cover}. 
Recently, the parametrized entropy has gained a great deal of attention in physics and information literature \cite{Tent}, 
in an effort to gain deeper understanding of the structure of equilibrium statistical mechanics and improved perspective 
of information theory. Due to the close connection to entropy and relative entropy, a generalized Kullback-Leibler (KL) 
relative entropy was presented, whose form is in conformity with a generalized entropy \cite{Borl,Tsallis}.  
In terms of the information gain, the generalization of the KL relative entropy has led to an adoption of a 
generalized information content.\\

The purpose of this paper is to introduce another extended KL divergence and to investigate its 
fundamental properties. Our construction of the new divergence measure is clear and origin of the 
form of the already existing generalization \cite{Borl,Tsallis} can be explained, once we realize that a seed 
quantity is a generalized overlap of the two distributions. The fundamental properties of the newly introduced 
generalization of KL that we treat in this paper include positivity, metricity, form invariance, concavity, upper and 
lower bounds and stability. In addition, the quantum no-cloning theorem \cite{No} has recently been shown to possess a 
classical counterpart, where universal perfect cloning machines are incompatible with the conservation of distance measure under 
the Liouville dynamics governing the evolution of the statistical ensemble. This fact was first shown through the ordinary 
KL divergence \cite{Daff} and was afterward extended to the Csisz{\' a}r f-divergence \cite{Plastino}. 
Furthermore it was shown that this fact also applies to a non f-divergence type \cite{Yamano1}. As a particular 
instance of the f-divergence, we shall specifically show that our generalized measure also exhibits constancy 
under this linear evolution dynamics. \\

The organization of the paper is as follows. In section II, we first recapitulate the basic properties of 
two different quantities: a distance in terms of KL, and the overlap of the distributions. We examine a specific 
stability property in order to clarify the difference between KL and the overlap. This property will be investigated 
for our measure in the subsequent section. In section III, we present our generalization by way of the asymmetric Jackson 
derivative. Some basic properties are addressed in section IV. We summarize our results in the final section.
\section{Distance and overlap}
The discrimination between two different probability functions is important in physics and information theory. 
To gain more insights into how we can measure the difference, we first consider the relationship between the 
KL distance and the overlap, which is also called the fidelity. Suppose that $n$ statistically independent subsystems 
constitute a system so that the joint probability distribution can be written in a factorized form 
$\mathcal{P}_m=\mathcal{P}_m^{(1)}\cdots \mathcal{P}_m^{(n)}$, where $m=1,2$. 
The KL distance $\mathcal{K}(\mathcal{P}_1,\mathcal{P}_2)$ between $\mathcal{P}_1$ and $\mathcal{P}_2$ 
defined on continuous support with $d{\bf x}=dx^{(1)}\cdots dx^{(n)}$ is 
\begin{eqnarray}
\int d{\bf x} \mathcal{P}_1\ln\frac{\mathcal{P}_1}{\mathcal{P}_2}&=& 
\int d{\bf x}\mathcal{P}_1^{(1)}\cdots \mathcal{P}_1^{(n)}\left(
\ln\frac{\mathcal{P}_1^{(1)}}{\mathcal{P}_2^{(1)}}+\cdots +
\ln\frac{\mathcal{P}_1^{(n)}}{\mathcal{P}_2^{(n)}}\right)\nonumber\\
&=&\sum_{j=1}^n\mathcal{K}(\mathcal{P}_1^{(j)},\mathcal{P}_2^{(j)}).
\end{eqnarray}
On the other hand, the overlap $\mathcal{O}(\mathcal{P}_1,\mathcal{P}_2)$ between 
$\mathcal{P}_1$ and $\mathcal{P}_2$, is comprised of the overlaps of the subsystems and is expressed as
\begin{eqnarray}
\int d{\bf x}\sqrt{\mathcal{P}_1\mathcal{P}_2} &=& \int d{\bf x}
\sqrt{\mathcal{P}_1^{(1)}\cdots \mathcal{P}_1^{(n)}\cdot 
\mathcal{P}_2^{(1)}\cdots \mathcal{P}_2^{(n)}}\nonumber\\
&=& \prod_{j=1}^n \mathcal{O}(\mathcal{P}_1^{(j)},\mathcal{P}_2^{(j)}). 
\end{eqnarray}
The KL distance is a sum of distances of the independent component systems 
(decomposability property), while the overlap for the total system is constructed from a product of 
the overlaps of subsystems. We note that as a closely related measure to the overlap, the statistical 
distance has been introduced by \cite{Woo}, and is based on the number of distinguishable states between 
two probabilities and can be given as an inverse cosine of the overlap as $\cos^{-1}\mathcal{O}(\mathcal{P}_1,\mathcal{P}_2)$.  
We used the continuous form in the above, however subtleties exist between the continuous and discrete form of entropies 
as clearly stated in \cite{Jumarie}. We shall use both forms depending on the ease of presentation in the rest of this paper. 
\subsection{Stability}\label{stab}
Stability in general has a quite broad meaning and has various definitions. The stability depends on the  
degrees of responses to the external perturbation. We shall now consider a situation where an external environment 
disturbs a system described by a set of probabilities. As a consequence of the disturbance, only a specific state 
of the system may slightly change the probability, say, by a factor $\epsilon$. Alternatively, we could also 
describe our set up as follows. The fluctuation of the target system is so small that its influence on the probability states 
could be limited and it appears only between two states. Due to the normalization of the entire probability, 
if the probability of a certain state is altered, then another state is also changed. 
This may be a matter of time scales inherent to the system. Although the fluctuation may initially occur locally in 
states space, it propagates in the neighboring states, and the reconfiguration of the probability distribution of the 
system occurs immediately towards a (quasi) equilibrium or static states of the system. Instead of considering 
the long-term stability, we limit our concern to a very early stage of the response. The long term dynamical stability 
requires the introduction of an underlying physics and is out of scope of the present treatment.\\

We define $\{p(x_i,\epsilon)\}$ as the distributions after an infinitesimal change denoted by a factor 
$\epsilon$, which is assumed to be close to unity. The evolution of the system is attributed to the change of the 
probabilities in time. Hence, two distributions $\{p(x_i)\}$ and $\{p(y_i)\}$ are assumed to be connected by 
$p(y_i)=\sum_{k=1}^np(y_i|x_k)p(x_k)$, i.e., the linear transformation of one into another state \cite{Vedral}. 
We may also regard this in the context of information theory as a transmission of input states 
$p(x_i)$'s under the channel matrix $p(y_i|x_k)$ to obtain output states $p(y_i)$. Then, the output states 
that received the disturbance are expressed as $p(y_i;\epsilon)=\sum_{k=1}^np(y_i|x_k)p(x_k;\epsilon)$. 
Without loss of generality, we assume that the fluctuation affects only two states ($l$th and $m$th) in 
such a way that \cite{ty}
\begin{equation}
p(x_k;\epsilon)=\left\{
\begin{array}{rl}
\epsilon p(x_l) & \mbox{for $k=l$} \\
(1-\epsilon)p(x_l)+p(x_m) & \mbox{for $k=m$}\\
p(x_k) & \mbox{for others}
\end{array}
\right..\label{eqn:pdef}
\end{equation}
This appears to correspond to the situation that a certain external fluctuation boosts the 
visiting frequency of a particular state. 
From the above, we have $p(y_i;\epsilon)=p(y_i)+c_i(\epsilon -1)p(x_l)$, where we have set 
$c_i=p(y_i\mid x_l)-p(y_i\mid x_m)$. Let us put $\epsilon-1=\xi$. We then have 
\begin{equation}
\mathcal{K}(\{p(y_i;\xi)\},\{p(y_i)\})=\sum_{i=1}^n\left[p(y_i)+\xi c_ip(x_l)\right]
\ln\left[1+\xi c_i\frac{p(x_l)}{p(y_i)}\right].
\end{equation}
Since we are considering the case $\xi\ll1$, then expanding the logarithm and considering above 
up to the second order in $\xi$, we have 
\begin{equation}
\mathcal{K}(\{p(y_i;\xi)\},\{p(y_i)\})=(\sum_{i=1}^nc_i)p(x_l)\xi+
\left(\sum_{i=1}^n\frac{c_i^2}{p(y_i)}\right)\frac{p^2(x_l)}{2}\xi^2+O(\xi^3).\label{eqn:KLs}
\end{equation} 
Therefore, we have always the positive second derivative $\partial^2\mathcal{K}/\partial \xi^2>0$, 
which means that the distance is stable under this disturbance.
On the other hand, the overlap between $p(y_i;\epsilon)$ and $p(y_i)$ are calculated to be 
\begin{eqnarray}
\mathcal{O}(\{p(y_i;\xi)\},\{p(y_i)\}) &=& \sum_{i=1}^n \sqrt{p(y_i;\xi)p(y_i)}\nonumber\\
&=& \sum_{i=1}^n p(y_i)\sqrt{1+\xi c_i\frac{p(x_l)}{p(y_i)}}\nonumber\\
&\sim& \sum_{i=1}^n p(y_i)\left[ 1+\frac{1}{2}c_i\frac{p(x_l)}{p(y_i)}\xi-
\frac{1}{8}c_i^2\left(\frac{p(x_l)}{p(y_i)}\right)^2\xi^2\right], 
\end{eqnarray}
where we have approximated the last line using $\sqrt{1+a\xi}\sim 1+a\xi/2-a^2\xi^2/8+\cdots$. 
Therefore, the second derivative $\partial^2 \mathcal{O}/\partial \xi^2=-p(x_l)^2(\sum_ic_i^2/p(y_i))/4<0$, 
which implies instability in this framework. Note that for two identical distributions, the KL distance vanishes while the 
overlap becomes unity. Therefore, the impact of the fluctuating effect on the two distance measures appears 
in the coefficients of $\xi^n (n>1)$.     
\section{A generalized Kullback-Leibler entropy}
The relative entropy can be arbitrarily defined and therefore it is possible to introduce alternative 
definitions to the conventional KL if needed. Some classes are actually discussed in \cite{Kapur}. 
Although the extensions of the usual KL entropy were already proposed by several authors in different 
forms, their presentation are somewhat heuristic and the mathematical origins are not fully clear. 
In this section, we consider a generalization of the KL entropy in light of the Jackson derivative \cite{Jack} 
and illustrate some of its properties in the next section. The Jackson derivative has its root in quantum group 
theory and has already been used to produce the Tsallis type generalized entropy \cite{Abe}. The Jackson derivative 
of a function $f(x)$ is defined for $s\ne 1$ by 
\begin{equation}
\frac{d}{d_s \alpha}f(\alpha):=\frac{f(s\alpha)-f(\alpha)}{s\alpha -\alpha}.
\end{equation}
The case $s=1$ corresponds to the ordinary derivative. The generalized entropy of Tsallis \cite{Tsallis} is obtained 
when the derivative is operated to a quantity $Z(\alpha)=\sum_ip_i^\alpha$ and evaluated at $\alpha=1$ \cite{Abe}, i.e., 
\begin{eqnarray}
\frac{d}{d_s\alpha}Z(\alpha)\bigg|_{\alpha=1}=\frac{1-\sum_i{p_i^s}}{s-1}.
\end{eqnarray} 
Keep in mind that if we operate the derivative $d/d_s \alpha$ on another quantity 
and evaluate it at different values of $\alpha$, we can in principle obtain different types of 
generalized entropies. In this scheme, we shall employ a quantity 
$Z^\prime(\alpha)=\sum_ip_i^\alpha q_i^{1-\alpha}$ for obtaining a new class of generalized KL entropy.
This quantity was called the R{\'e}nyi overlap of order $\alpha$ \cite{Fuchs} because 
a quantity $\ln Z(\alpha)/(\alpha-1)$ is defined by R{\'e}nyi \cite{Renyi,Renyi2}. 
We note that the usual KL entropy is obtained by the ordinary derivative of $Z^\prime(\alpha)$ 
when evaluated either at $\alpha=1$,  
\begin{eqnarray}
\frac{d Z^\prime (\alpha)}{d_1\alpha}\bigg|_{\alpha=1}
=\mathcal{K}(\mathcal{P},\mathcal{Q}),
\end{eqnarray}
or at $\alpha=0$,
\begin{eqnarray}
\frac{d Z^\prime (\alpha)}{d_1\alpha}\bigg|_{\alpha=1}
=-\mathcal{K}(\mathcal{Q},\mathcal{P}).
\end{eqnarray}
By the same token, the generalized KL entropy introduced previously in \cite{Tsallis,Borl} is generated by 
the operation of $d/d_s\alpha$ to $Z^\prime(\alpha)$ and substituting $\alpha=1$, 
\begin{eqnarray}
\mathcal{K}_s(\mathcal{P},\mathcal{Q})=\frac{dZ^\prime (\alpha)}{d_s\alpha}\bigg|_{\alpha=1}
=\frac{\sum_ip_i^sq_i^{1-s}-1}{s-1}.
\end{eqnarray}
These facts provide an indication that we can produce various kinds of generalized KL entropies 
by evaluating the Jackson derivative of $Z^\prime(\alpha)$ at different values of $\alpha$. 
It is also possible to take another approach to achieve a generalization by employing the 
symmetric Jackson derivative defined for a function $g(\alpha)$ with $s\ne 1$
\begin{eqnarray}
\mathscr{D}_{s;\alpha}[g(\alpha)]:=\frac{g(s\alpha)-g(s^{-1}\alpha)}{(s-s^{-1})\alpha}.
\end{eqnarray}
This derivative is symmetric under the interchange of $s\leftrightarrow s^{-1}$. 
We operate the symmetric Jackson derivative on $Z^\prime(\alpha)$ and evaluate it at $\alpha=1$,
\begin{eqnarray}
\mathcal{L}_s(\mathcal{P},\mathcal{Q}) &:=& \mathscr{D}_{s;\alpha}\left[\sum_ip_i^\alpha q_i^{1-\alpha}\right]
\bigg|_{\alpha=1}\nonumber\\
&=& \frac{1}{s-s^{-1}}\sum_ip_i\left[ \left(\frac{q_i}{p_i}\right)^{1-s}- 
\left(\frac{q_i}{p_i}\right)^{1-s^{-1}}\right].
\end{eqnarray}
We note that this generalized KL entropy is asymmetric $\mathcal{L}_s(\mathcal{P},\mathcal{Q})\neq 
\mathcal{L}_s(\mathcal{Q},\mathcal{P})$ and the relation $\mathcal{L}_{s^{-1}}(\mathcal{P},\mathcal{Q})
=\mathcal{L}_s(\mathcal{P},\mathcal{Q})$ holds. A symmetric quantity can be constructed by adding  
$\mathcal{L}_s(\mathcal{P},\mathcal{Q})$ and $\mathcal{L}_s(\mathcal{Q},\mathcal{P})$.  
In the limit $s\rightarrow 1$, $\mathcal{L}_s(\mathcal{P},\mathcal{Q})$ reduces to the usual KL entropy 
$\mathcal{K}(\mathcal{P},\mathcal{Q})$, which can be easily checked by the L'Hospital theorem. 
This divergence is well-defined whenever the two distributions have common support (state number i's). 
In other words, in order to have a finite value as a distance measure in the case $0<s<1$, the probability $p_i$ 
must vanish when $q_i$ vanishes, and similar restrictions also apply for $s>1$ and in the limit 
$s\to 1$.  
The decomposability property in the sense that we mentioned in section II is not expected for $\mathcal{L}_s$, 
since 
\begin{equation}
(s-s^{-1})\mathcal{L}_s(\mathcal{P}_1^{(1)}\cdots \mathcal{P}_1^{(n)},\mathcal{P}_2^{(1)}\cdots \mathcal{P}_2^{(n)})
=\prod_{j=1}^n L_s^{(j)}-\prod_{j=1}^n L_{s^{-1}}^{(j)}, 
\end{equation}
where $L_s^{(j)}=\int d{\bf x}^{(j)}[\mathcal{P}_1^{(j)}]^s[\mathcal{P}_2^{(j)}]^{1-s}$ etc. 
\cite{cmt}.
It would be interesting to note that $\mathcal{L}_s$ can be understood from the limiting case of 
the weighted power mean of order $\lambda$, which is defined for $x,y>0$ as 
\begin{eqnarray}
E_s^\lambda[x,y]:=[sx^\lambda+(1-s)y^\lambda]^\frac{1}{\lambda},
\end{eqnarray}
and the particular instances $E_{\frac{1}{2}}^0[x,y]$ and $E_{\frac{1}{2}}^1[x,y]$ 
correspond to the geometric mean $\sqrt{xy}$ and to the arithmetic mean $(x+y)/2$, respectively. Therefore we have
\begin{eqnarray}
\mathcal{L}_s(\mathcal{P},\mathcal{Q})=\lim_{\lambda\to 0}\frac{E_s^\lambda[p_i,q_i]
-E_{s^{-1}}^\lambda[p_i,q_i]}{s-s^{-1}}.
\end{eqnarray}
\section{Some properties of $\mathcal{L}_s(\mathcal{P},\mathcal{Q})$}
\subsection{Positive semi-definiteness}\label{posit}
This property corresponds to the information inequality for the standard KL entropy, i.e., 
$\mathcal{K}(\mathcal{P},\mathcal{Q})\geqslant 0$. 
For $\mathcal{L}_s(\mathcal{P},\mathcal{Q})$, the kernel function is   
$f(x)=(x^{1-s}-x^{1-s^{-1}})/(s-s^{-1})$. The second derivative 
\begin{eqnarray}
f^{''}(x)=\frac{1}{s-s^{-1}}\left\{ s(s-1)x^{-1-s}+\frac{1}{s}(1-\frac{1}{s})x^{-1-\frac{1}{s}}\right\}
\end{eqnarray}
is always positive for $s>0$, zero for $s=0$ and can be negative for $s<0$. 
Therefore, due to the Jensen's inequality $\sum_i\alpha_if(x_i)\gtreqless f(\sum_i\alpha_ix_i)$ for 
$f^{''}(x)\gtreqless 0$ with $\sum_i\alpha_i=1$, by putting $x_i=q_i/p_i$ we obtain 
\begin{eqnarray}
\mathcal{L}_s(\mathcal{P},\mathcal{Q})&\gtreqless& \frac{1}{s-s^{-1}}\left\{ 
\left[ \sum_ip_i\left( \frac{q_i}{p_i}\right)\right]^{1-s}-
\left[ \sum_ip_i\left( \frac{q_i}{p_i}\right)\right]^{1-s^{-1}}\right\}\nonumber\\
&=& 0.
\end{eqnarray}
Note that for $s\neq 0$, the last equality holds iff $p_i=q_i$, $\forall i$. 
We note that when $s>0$, $f$ is a convex function. Accordingly the positivity $\mathcal{L}_s\geqslant 0$ 
is found to be a direct consequence of Lemma 1.1 in \cite{Csis}, where an inequality 
\begin{eqnarray}
\int_{E}p(x)f\left( \frac{q(x)}{p(x)}\right)m(dx)\geqslant \int_{E}p(x)m(dx)\cdot f(u_0)\label{eqn:Csis}
\end{eqnarray}
is proved for nonnegative measurable functions on a measure space $(X,\mathfrak{X},m)$ for 
$E\in \mathfrak{X}$ ($\sigma$-algebra of subsets of $X$) and $u_0=\int_{E}q(x)m(dx)/\int_{E}p(x)m(dx)$. 
The normalization of the probability functions in Eq.(\ref{eqn:Csis}) when $f $ satisfies $f(1)=0$ results in 
positivity for our case.
\subsection{Metric property}
The infinitesimal shift in probability provides 
\begin{eqnarray}
\mathcal{L}_s(\mathcal{P},\mathcal{P}+d\mathcal{P})=\sum_ip_if\left( 1+\frac{\delta p_i}{p_i}\right)
\approx \sum_{n=2}^\infty\frac{f^{(n)}(1)}{n!}\sum_i\frac{(\delta p_i)^n}{p_i^{n-1}},
\end{eqnarray}
where $f(x)$ is the same as above, and we have used the facts $f(1)=0$ and $\sum_i\delta p_i=0$. 
Since the second derivative is $f^{''}(1)=[s(s-1)+s^{-1}(1-s^{-1})]/(s-s^{-1})$ 
and we have $\mathcal{L}_s(\mathcal{P},\mathcal{P}+d\mathcal{P})\sim 2^{-1}f^{''}(1)\sum_i(\delta p_i)^2/p_i$, 
if we introduce the information metric (or Fisher-Rao metric) $ds^2$ as
\begin{eqnarray}
ds^2=\sum_{i,j}g_{ij}dp^idp^j,
\end{eqnarray}
then the metric tensor $g_{ij}$ is given by
\begin{eqnarray}
g_{ij}=\frac{s(s-1)+s^{-1}(1-s^{-1})}{2p_i(s-s^{-1})}\delta_{ij}.
\end{eqnarray}
In the limit $s\to 1$, this metric reduces to $\delta_{ij}/(2p_i)$. 
\subsection{Form invariance}
Under transformation of coordinates $\theta\to \eta$, the distribution $p(\theta)$ may satisfy   
$p(\theta)d\theta=\phi(\eta)d\eta$ where $\phi(\eta)$ is the converted distribution. 
The resulting relation $p_1/p_2=\phi_1/\phi_2$ shows that the distance between $\phi_1$ and $\phi_2$ 
measured by $\mathcal{L}_s$ remains unchanged, that is, the distance is equivalent to 
the one between $p_1$ and $p_2$ before the transformation,
\begin{eqnarray}
\mathcal{L}_s(\mathcal{P}_1,\mathcal{P}_2)=\frac{1}{s-s^{-1}}\int d\eta \phi_1(\theta(\eta)) 
\left[\left(\frac{\phi_2(\theta(\eta))}{\phi_1(\theta(\eta))}\right)^{1-s}
-\left(\frac{\phi_2(\theta(\eta))}{\phi_1(\theta(\eta))}\right)^{1-s^{-1}}\right]
=\mathcal{L}_s(\Phi_1,\Phi_2).
\end{eqnarray}
\subsection{Concavity}
By setting $\alpha_j=a_j/\sum_k a_k$ and $x_j=b_j/a_j$ for the 
Jensen's inequality $\sum_j \alpha_jf(x_j)\geqslant f(\sum_j\alpha_j x_j)$ with the same $f(x)$ as in 
the section \ref{posit}, we have 
\begin{eqnarray}
\frac{1}{\sum_k a_k}\frac{\sum_j a_j\left[\left(\frac{b_j}{a_j}\right)^{1-s}-
\left(\frac{b_j}{a_j}\right)^{1-s^{-1}}\right]}{s-s^{-1}}
\geqslant 
\frac{\left[\left(\frac{\sum_jb_j}{\sum_ka_k}\right)^{1-s}-
\left(\frac{\sum_jb_j}{\sum_ka_k}\right)^{1-s^{-1}}\right]}{s-s^{-1}}.\label{eqn:Lqineq}
\end{eqnarray}
We consider two states $j=1,2$ by putting $a_1=\lambda p_j^1$ and $a_2=(1-\lambda)p^2_j$ 
to obtain $p=\lambda p_j^1+(1-\lambda)p^2_j$. 
Similarly, we set $b_1=\lambda p_j^{\prime 1}$ and $b_2=(1-\lambda)p_j^{\prime 2}$ for 
$p^\prime=\lambda p_j^{\prime 1}+(1-\lambda)p_j^{\prime 2}$, where $0\leqslant\lambda\leqslant 1$. 
Substituting these into Eq.(\ref{eqn:Lqineq}), and summing over $j$ yields
\begin{eqnarray}
\lambda \mathcal{L}_s(p^1,p^{\prime 1})+(1-\lambda)\mathcal{L}_s(p^2,p^{\prime 2})
\geqslant
\mathcal{L}_s(\lambda p^1+(1-\lambda)p^2,\lambda p^{\prime 1}+(1-\lambda)p^{\prime 2}).
\end{eqnarray}
This completes the proof.  
\subsection{Stability}\label{stab2}
We investigate the stability property explained in section \ref{stab}. The distance measure between 
the original and the perturbed distribution is 
\begin{eqnarray}
\mathcal{L}_s(\{p(y_i;\xi)\},\{p(y_i)\})&=&\frac{1}{s-s^{-1}}\sum_{i=1}^n\left\{p(y_i)+c_ip(x_l)\xi\right\}\nonumber\\
&\times& \left[\left(\frac{p(y_i)+c_ip(x_l)\xi}{p(y_i)}\right)^{s-1}-
\left(\frac{p(y_i)+c_ip(x_l)\xi}{p(y_i)}\right)^{s^{-1}-1}\right].
\end{eqnarray}
Expanding $(1+\xi c_ip(x_l)/p(y_i))^{s-1}$ etc. with respect to $\xi$ and taking terms up to second 
order in $\xi$, we obtain the expression
\begin{eqnarray}
\mathcal{L}_s(\{p(y_i;\xi)\},\{p(y_i)\})=(\sum_{i=1}^nc_i)p(x_l)\xi+
\frac{f(s)}{2}\left(\sum_{i=1}^n\frac{c_i^2}{p(y_i)}\right)p^2(x_l)\xi^2+O(\xi^3),
\end{eqnarray}
where $f(s)=s-s^{-1}-1$. Considering the sign of the coefficient of $\xi^2$, we can conclude that the 
$\mathcal{L}_s$ is stable when $s>0$ and unstable when $s<0$. Note that except for the factor $f(s)$, 
the effect of perturbation on the generalized distance is the same as on the 
ordinary KL Eq.(\ref{eqn:KLs}) up to second order in $\xi$. 
\subsection{Upper and lower bounds for $\mathcal{L}_s(\mathcal{P},\mathcal{Q})$}
The following bounds hold.
\begin{Theorem}
Let $p(x)$, $q(x)\in \mathcal{X}$, be two probability distributions. Then we have the inequality:
\begin{eqnarray}\label{eq1}
\frac{1}{2}\sum_{x\in\mathcal{X}}\left[\left(\frac{q(x)}{p(x)}\right)^{1-s}+
\left(\frac{q(x)}{p(x)}\right)^{1-s^{-1}}\right] p(x)\log \frac{p(x)}{q(x)}
&\leqslant& 
\mathcal{L}_s(\mathcal{P},\mathcal{Q})\nonumber\\
&\leqslant& 
\sum_{x\in\mathcal{X}}\left(\frac{q(x)}{p(x)}\right)^{1-\frac{s+s^{-1}}{2}}\!\!\!\!\!p(x)\log \frac{p(x)}{q(x)},
\end{eqnarray}
\end{Theorem}

\begin{proof}
For a convex function $f(u)=t^{1-u}$ ($0<t<1$), we employ the Hermite-Hadamard inequality which holds 
for convex functions,
\begin{equation}\label{pr1-eq1}
f(\frac{a+b}{2})\leqslant \frac{1}{b-a}\int^{b}_{a}f(u)du\leqslant \frac{f(a)+f(b)}{2},\quad (a,b >0).
\end{equation}
Now, putting $a=s$ and $b=s^{-1}$ we have an inequality
\begin{equation}
t^{1-\frac{s+s^{-1}}{2}}\leqslant \frac{t^{1-s}-t^{1-s^{-1}}}{(s-s^{-1})\log t}\leqslant 
\frac{1}{2}(t^{1-s}+t^{1-s^{-1}}).\label{eqn:pr1-eq2}
\end{equation}
Since $(t^{1-s}-t^{1-s^{-1}})/(s-s^{-1})$ takes negative values $\forall s$ for $t\in (0,1)$,
if we choose $t=q(x)/p(x)$ in equation (\ref{eqn:pr1-eq2}) and sum over $x\in \mathcal{X}$ after 
multiplying by $p(x)$, we obtain 
\begin{eqnarray}
\frac{1}{2}\sum_{x\in\mathcal{X}}\left[\left(\frac{q(x)}{p(x)}\right)^{1-s}+
\left(\frac{q(x)}{p(x)}\right)^{1-s^{-1}}\right]p(x)\log \frac{p(x)}{q(x)} 
&\leqslant& 
\frac{1}{s-s^{-1}}\sum_{x\in\mathcal{X}}p(x)\left[\left(\frac{q(x)}{p(x)}\right)^{1-s}-
\left(\frac{q(x)}{p(x)}\right)^{1-s^{-1}}\right]\nonumber\\
&\leqslant& 
\sum_{x\in\mathcal{X}}\left(\frac{q(x)}{p(x)}\right)^{1-\frac{s+s^{-1}}{2}}\!\!\!\!p(x)\log \frac{p(x)}{q(x)}.
\end{eqnarray}
The equality holds if and only if $s=s^{-1}$, i.e., $s=\pm 1$, which completes the desired bounds for $\mathcal{L}_s$.
\end{proof}
As for an upper bound for $\mathcal{L}_s$, the following expression also holds.
\begin{Corollary}{}
Let $p(x)$, $q(x)$ be as above. For $-1\leqslant s \leqslant 0$ and $\quad s>1$, we have
\begin{equation}
0\leqslant \mathcal{L}_s(\mathcal{P},\mathcal{Q})\leqslant \frac{1}{s-s^{-1}}\sum_{x\in\mathcal{X}}\frac{|\exp[p^s(x)q^{1-s}(x)]-
\exp[p^{s^{-1}}(x)q^{1-s^{-1}}(x)]|}{\sqrt{\exp[p^s(x)q^{1-s}(x)+p^{s^{-1}}(x)q^{1-s^{-1}}(x)]}}.\label{eqn:coro1}
\end{equation}
\end{Corollary}
\begin{proof}
For $a\geqslant 0$ and $b\geqslant 0$, the geometric mean is smaller than or equal to the logarithmic mean, i.e., 
\begin{eqnarray}
\sqrt{ab}\leqslant \frac{b-a}{\ln b-\ln a}.
\end{eqnarray}
Equivalently the inequality $|\ln b-\ln a|\leqslant |b-a|/\sqrt{ab}$ holds for the equality iff $a=b$. Therefore, 
setting $b=\exp[p^s(x)q^{1-s}(x)]$ and $a=\exp[p^{s^{-1}}(x)q^{1-s^{-1}}(x)]$, we have
\begin{eqnarray}
\bigg|p^s(x)q^{1-s}(x)-p^{s^{-1}}(x)q^{1-s^{-1}}(x)\bigg|\leqslant\frac{|e^{p^s(x)q^{1-s}(x)}-e^{p^{s^{-1}}(x)q^{1-s^{-1}}(x)}|}
{\sqrt{\exp[p^s(x)q^{1-s}(x)+p^{s^{-1}}(x)q^{1-s^{-1}}(x)}]}.\label{eqn:T1}
\end{eqnarray}
From the positivity property proved in section 5.1, we have $0\leqslant \mathcal{L}_s(\mathcal{P},\mathcal{Q})$. Then 
\begin{eqnarray}
(s-s^{-1})\mathcal{L}_s(\mathcal{P},\mathcal{Q})&=&
\bigg|\sum_{x\in\mathcal{X}}p^s(x)q^{1-s}(x)-p^{s^{-1}}(x)q^{1-s^{-1}}(x)\bigg|\nonumber\\
&\leqslant& \sum_{x\in\mathcal{X}}\bigg|p^s(x)q^{1-s}(x)-p^{s^{-1}}(x)q^{1-s^{-1}}(x)\bigg|.
\end{eqnarray}
Summing over $x$ in Eq.(\ref{eqn:T1}), we obtain the inequality Eq.(\ref{eqn:coro1}).
\end{proof}
Another expression for the bound in terms of the $l$-norm is possible for the exponentiated differences.
\begin{Corollary}
Let $p(x)$, $q(x)$ be as above. Then we have 
\begin{equation}
0\leqslant \mathcal{L}_s(\mathcal{P},\mathcal{Q})\leqslant 
\Bigm\| e^{p^s(x)q^{1-s}(x)}-e^{p^{s^{-1}}(x)q^{1-s^{-1}}(x)}\Bigm\|_\alpha\cdot
\Bigm\|\exp[-p^s(x)q^{1-s}(x)-p^{s^{-1}}(x)q^{1-s^{-1}}(x)]\Bigm\|_{\frac{\beta}{2}}^{\frac{1}{2}}
\end{equation}
where $||t||_l:=\left[\sum_{x\in \mathcal{X}}|t(x)|^l\right]^{1/l}$ for $l>0$. 
\end{Corollary} 
\begin{proof}
From the H{\" o}lder's inequality with $1/\alpha+1/\beta=1$, it immediately follows that
\begin{eqnarray}
& & \sum_{x\in\mathcal{X}} \frac{|e^{p^s(x)q^{1-s}(x)}-e^{p^{s^{-1}}(x)q^{1-s^{-1}}(x)}|}{\sqrt{\exp[p^s(x)q^{1-s}(x)+p^{s^{-1}}(x)q^{1-s^{-1}}(x)}]}\nonumber\\
&\leqslant& \left[\sum_{x\in\mathcal{X}}| e^{p^s(x)q^{1-s}(x)}-e^{p^{s^{-1}}(x)q^{1-s^{-1}}(x)}|^\alpha \right]^{1/\alpha}
\left[\sum_{x\in\mathcal{X}} \left(\frac{1}{\sqrt{\exp[p^s(x)q^{1-s}(x)+p^{s^{-1}}(x)q^{1-s^{-1}}(x)}]}\right)^\beta\right]^{1/\beta}.
\end{eqnarray}
\end{proof}
\section{Shift information for $\mathcal{L}_s(\mathcal{P},\mathcal{Q})$}
The notion of the shift information introduced in Ref.\cite{Vst} is an interesting means of investigating 
our new distance measure, in that we may ascribe the infinitesimal shift to 
known quantities. The original definition of the shift information can be expressed by using the ordinary KL 
entropy as $\mathcal{K}(p(\zeta),p(\zeta+\Delta))$, where the $\Delta$ is a sufficiently small quantity compared to 
the variable $\zeta$. As a consequence of the expansion, the Fisher information measure appears in 
the second order term in the case of the usual KL entropy \cite{Vst} and also in a generalized KL 
entropy \cite{Borl}. We would be able to expect that the shift information for the present generalization can also be 
expressible in terms of the Fisher information, where the generalization parameter $s$ should 
govern the degree of the shift. We look this fact below. The shift information is defined as 
\begin{eqnarray}
I(\Delta,s):=\mathcal{L}_s(p(\zeta),p(\zeta+\Delta))=\frac{1}{s-s^{-1}}\int d\zeta  
\left[ p^s(\zeta)[p(\zeta+\Delta)]^{1-s}-p^{s^{-1}}(\zeta)[p(\zeta+\Delta)]^{1-s^{-1}}\right].
\end{eqnarray}
Expanding $[p(\zeta+\Delta)]^{1-\gamma}$ with respect to $\Delta$, 
\begin{eqnarray}
p^\gamma(\zeta)[p(\zeta+\Delta)]^{1-\gamma}\sim p(\zeta)+(1-\gamma)p^\prime(\zeta)\Delta+
(1-\gamma)\left\{p^{''}(\zeta)-\gamma\frac{(p^\prime(\zeta))^2}{p(\zeta)}\right\}\frac{\Delta^2}{2}+\cdots,
\end{eqnarray}
where $\gamma$ denotes either $s$ or $s^{-1}$ and the prime implies the derivative with respect to $\zeta$.  
Then the shift information can be expressed up to the second order in $\Delta$ as 
\begin{eqnarray}
I(\Delta,s)=\frac{1}{s-s^{-1}}\int d\zeta\left\{ \left(s^{-1}-s\right)p^\prime(\zeta)\Delta+
\left[(s^{-1}-s)p^{''}(\zeta)+a(s)\frac{[p^\prime(\zeta)]^2}{p(\zeta)}\right]
\frac{\Delta^2}{2}\right\},\label{eqn:dIs}
\end{eqnarray}
where we have put $a(s)=s(s-1)+s^{-1}(1-s^{-1})$. Therefore we find that the Fisher information 
$\int d\zeta (p^\prime)^2/p$ is a relevant quantity to the second order in the shift $\Delta$.  
Moreover, the variation of $I(\Delta,s)$ is given by  
\begin{eqnarray}
\delta I(\Delta,s) &=&\frac{1}{s-s^{-1}}\int d\zeta\delta p\left[\frac{\partial}{\partial p}- 
\frac{\partial}{\partial\zeta}\frac{\partial}{\partial p^\prime}+\frac{\partial^2}{\partial\zeta^2}
\frac{\partial}{\partial p^{''}}-\cdots\right]\times\nonumber\\
&&\left\{ (s^{-1}-s)p^\prime\Delta+[(s^{-1}-s)p^{''}+a(s)\frac{(p^\prime)^2}{p}]\frac{\Delta^2}{2}\right\}
\nonumber\\
&\simeq &\frac{a(s)}{2(s-s^{-1})}\delta I_{KL}(\Delta),
\end{eqnarray}
where $\delta I_{KL}(\Delta)$ is the variation calculated for the shift information for the ordinary 
KL \cite{Vst}, 
\begin{equation}
\int d\zeta\delta p\left\{\left(\frac{p^\prime}{p}\right)^2-2\frac{p^{''}}{p}\right\}.
\end{equation}
This result indicates that the variation simply differs by a factor from the one obtained for the ordinary KL, 
whose degree is controlled by the index $s$. If the second derivative of the distance measure indicates a direct 
quantity and is responsible for keeping the discernible interval between the shifted and the original, then the sign of 
$\partial^2 I(\Delta,s)/\partial \Delta^2$ would be a signature of this stability. As an example, consider a Gaussian form 
as a representative distribution which appears in many disciplines. The Fisher information is calculated to be 
$\sigma^{-2}$ for the domain $\zeta\in [-\infty,\infty]$, where $\sigma$ is the standard deviation of the 
Gaussian distribution function. By straightforward calculation, we obtain 
\begin{eqnarray}
\frac{\partial^2 I(\Delta,s)}{\partial \Delta^2}=\frac{a(s)}{\sigma^2 (s-s^{-1})}.
\end{eqnarray}
We find that $a(s)/(s-s^{-1}) \gtrless 0$ when $s\gtrless 0$, indicating that the information is stable against the 
shift $\Delta$ when $s>0$. We note that this result is consistent with the conclusion derived from the two level 
perturbation approach obtained in section \ref{stab2}, where the corresponding $\mathcal{L}_s$ is stable (unstable) 
if $s>0$ ($s<0$). In this sense, the two different approaches for investigation of the stability associated with the distinguishability 
can be regarded as equivalently informative. It is worth mentioning that in the case of the R{\'e}nyi relative entropy we 
obtain the shift information as,
\begin{eqnarray}
I^{R}(\Delta ,s)&:=&\frac{1}{s-1}\ln\left[\int d\zeta p(\zeta)\left(\frac{p(\zeta+\Delta)}{p(\zeta)}\right)^{1-s}\right]\nonumber\\
&\simeq & \frac{1}{s-1}\ln\left[1+(1-s)\Delta \int d\zeta p^\prime+
\frac{(1-s)}{2}\Delta^2\int d\zeta \left(p^{''}-\frac{(p^\prime)^2}{p}\right)\right].
\end{eqnarray}
For the Gaussian distribution, the second derivative is calculated as 
\begin{eqnarray}
\frac{\partial^2 I^{R}(\Delta,s)}{\partial \Delta^2}=\frac{s(s-1)(\Delta-1)^2+2\sigma^2-s(s-1)}{\left[2\sigma^2-s(1-s)\Delta^2\right]^2}.
\end{eqnarray}
Therefore, when $\sigma^2>s(s-1)/2$ is satisfied, $I^{R}(\Delta,s)$ is found to be stable. 
The remarkable difference between $I(\Delta,s)$ and $I^{R}(\Delta,s)$ is that the stability is controlled only by $s$ for 
our shift information, whereas the form of distribution (i.e., the magnitude of $\sigma$) imposes the restriction for the domain 
of $s$ in the case of the R{\'e}nyi shift information in general.
\section{Behavior under Liouville dynamics}
We shall prove in this section that two states can only become less distinguishable in the course of a dynamical 
evolution when the distance between them are measured by the present one. In other words, the generalized KL 
entropy $\mathcal{L}_s$ does not increase with time, instead is shown to be constant in time under the Liouville equation 
\begin{eqnarray}
\frac{\partial p}{\partial t}+\nabla\cdot(\vec{v}p)=0,\label{eqn:LE} 
\end{eqnarray}
where $p(\vec{\zeta},t)$ denotes a probability density describing a statistical ensemble of dynamical systems 
and $\vec{v}=d\vec{\zeta}/dt$ stands for the drift velocity. The time derivative of the generalized KL 
entropy for two arbitrary probability distributions which satisfy the Liouville equation is 
\begin{eqnarray}
\frac{d\mathcal{L}_s(\mathcal{P}_1,\mathcal{P}_2)}{dt} &=& \frac{1}{s-s^{-1}}\int d\zeta\frac{\partial}{\partial t}
\left[p_1^s p_2^{1-s}-p_1^{1/s}p_2^{1-1/s}\right]\nonumber\\
&=& \frac{1}{s-s^{-1}}\int d\zeta \left[ f_1(p_1,p_2)\left( \frac{\partial p_1}{\partial t}\right)
+f_2(p_1,p_2)\left( \frac{\partial p_2}{\partial t}\right)\right],\label{eqn:dLs}
\end{eqnarray}
where 
\begin{eqnarray}
f_1=s\left( \frac{p_2}{p_1}\right)^{1-s}-\frac{1}{s}\left( \frac{p_2}{p_1}\right)^{1-\frac{1}{s}}, \quad 
f_2=(1-s)\left( \frac{p_2}{p_1}\right)^{-s}-(1-\frac{1}{s})\left( \frac{p_2}{p_1}\right)^{-\frac{1}{s}}.
\end{eqnarray}
Substituting Eq.(\ref{eqn:LE}) into Eq.(\ref{eqn:dLs}), then using $f\nabla g=\nabla(fg)-(\nabla f)g$, 
we obtain for the integral of Eq.(\ref{eqn:dLs}),
\begin{eqnarray}
&&\int d\zeta\left[\left\{s(1-s)\left(\frac{p_2}{p_1}\right)^{-s}-\frac{1}{s}(1-\frac{1}{s})
\left(\frac{p_2}{p_1}\right)^{-\frac{1}{s}}\right\}p_1 \right.\nonumber\\
&+& \left.\left\{-s(1-s)\left(\frac{p_2}{p_1}\right)^{-s-1}-\frac{1}{s}(1-\frac{1}{s})
\left(\frac{p_2}{p_1}\right)^{-\frac{1}{s}-1}\right\}p_2\right]
\vec{v}\nabla\left(\frac{p_2}{p_1}\right),
\end{eqnarray}
where we have assumed that the two probability distributions $f$ and $g$ vanish at the boundary, so that 
$\int \nabla(fg)d\zeta=0$ holds. The quantity within the bracket is calculated to be zero, therefore 
$\mathcal{L}_s$ is found to be an invariant measure under this dynamics for all values of $s$. 
It it worth mentioning that relative entropies of the form 
$\int d\zeta \mathcal{P}_1f(\mathcal{P}_2/\mathcal{P}_1)$ (the Csisz{\'a}r f-divergence), where the 
function $f$ is convex and satisfies $f(1)=0$, becomes constant in time under the Liouville type dynamical 
evolution \cite{Daff,Plastino,Yamano1,Mackey}. The fact that $d\mathcal{L}_s/dt=0$ under the Liouville 
equation proved above is consistent with this observation because $\mathcal{L}_s$ is a particular 
instance of the Csisz{\'a}r f-divergence class.
\section{Conclusions}
We have investigated properties of a novel generalized KL divergence in the context of statistical physics and 
information theory. Our approach presents a unified recipe for constructing distance measures 
for probability distributions. In this method, the ordinary KL divergence is obtained by differentiation of the generalized 
overlap with respect to the overlap index $\alpha$ and evaluating it by its unity. Similarly, the previously 
reported generalization of the KL divergence, which is consistent with the nonextensive entropy proposed in 
physics literature, can be regarded as an output of  the Jackson derivative for the overlap 
evaluated by its unity. Along this line, we can define a family of distance measures by applying the symmetric 
Jackson derivative to the generalized overlap. We have chosen $\alpha=1$ to obtain a specific generalization, which belongs to 
the Csisz{\'a}r f-divergence type and have shown some fundamental properties of the divergence measure. 
As far as the distance between two probability distributions are concerned, the KL relative entropy has 
infinite generalizations even with our recipe, depending on the evaluation index. The connection to an interpretation 
of the information gain would provide the corresponding generalized information content. 
We have obtained the ratio of the variation of the shift information to that of the ordinary one
$\delta I(\Delta,s)/\delta I_{KL}(\Delta)$. In closing, we remark on a possible application of the divergence in the light of 
the minimum KL divergence scheme. In \cite{ARP}, this minimization formalism was applied to approximately obtain 
solutions of the general $N$-dimensional linear Fokker-Planck equations. Following this reasoning, the newly introduced divergence 
could be useful for finding approximate solutions to nonlinear Fokker-Planck equations and the related time evolution 
equations. Developing this approach would require future investigation.\\

{\it Acknowledgements}\\
This work was partially supported by the Grant-in-Aid for Scientific Research from Monbukagaku-sho No.06225 and 
was presented at the DPG (Deutsche Physikalisches Gesellschaft) conference as No. DY 30.7 at Regensburg University, 
26-30 March 2007.


\begin{thebibliography}{500}
\bibitem{cmt} $x_1x_2-y_1y_2$ is not expressible by $x_1-y_1$ and $x_2-y_2$ only.
\end{thebibliography}
\end{document}